%% file: StegoCodesArxiv.tex
\documentclass[runningheads,a4paper]{llncs}

\usepackage{caption}
\usepackage{fullpage}
\usepackage[cmex10]{amsmath}
\usepackage{amssymb}
\usepackage{amsfonts}
\usepackage{amscd}

\usepackage{array}
\usepackage{url}
\usepackage{graphicx}
\usepackage{mathtools}
\usepackage{etex}
\usepackage{multicol}
\usepackage{pgfplots}
\usetikzlibrary{matrix,positioning,decorations.pathreplacing,decorations.pathmorphing,backgrounds,chains}
\usepackage{relsize}

\colorlet{boxgray}{lightgray!20}

\usepackage[subpreambles]{standalone}
\usepackage{amssymb}
\usepackage{mathtools}
\usepackage{algorithm}
\usepackage{bm}
\usepackage{graphicx}
\usepackage{url}
\usepackage{tikz}
\usetikzlibrary{arrows,matrix,positioning}
\usepackage{algorithm}
\usepackage{algorithmic}
\usepackage{rotating}
\usepackage[english]{babel} 
\usepackage{color,array}
\usepackage{enumerate}
\usepackage{bm}
\usepackage{graphicx}
\usepackage[caption=false]{subfig}
\usepackage{url}
\usepackage{standalone}
\usepackage{tikz}
\usetikzlibrary{arrows,matrix,positioning,patterns}
\usepackage{pgfplots}
\usetikzlibrary{matrix,positioning,decorations.pathreplacing,decorations.pathmorphing,backgrounds,chains}
\usepackage{relsize}
\colorlet{boxgray}{lightgray!20}

\newcommand{\ie}{\emph{i.e.,\ }}

\DeclareMathOperator{\F}{\mathbb{F}}

\newcommand\Mark[1]{\textsuperscript{#1}}

\newsavebox{\ieeealgbox}

\begin{document}
\pagenumbering{gobble}
\title{Approaching Maximum Embedding Efficiency on Small Covers Using Staircase-Generator Codes}

\author{Simona Samardjiska\Mark{1,2}
and
Danilo Gligoroski\Mark{1}}

\authorrunning{S. Samardjiska and D. Gligoroski}

\institute{
Department of Telematics, 
Norwegian University of Science and Technology, 
Trondheim, Norway\Mark{1}\\ 
FCSE, ``Ss Cyril and Methodius'' University, Skopje, Republic of Macedonia\Mark{2}\\
\email{ \{simonas,danilog\}@item.ntnu.no}\\
}
\maketitle

\begin{abstract} We introduce a new family of binary linear codes suitable for steganographic matrix embedding. The main characteristic of the codes is the staircase random block structure of the generator matrix. We propose an efficient list decoding algorithm for the codes that finds a close codeword to a given random word. We provide both theoretical analysis of the performance and stability of the decoding algorithm, as well as practical results.  
Used for matrix embedding, these codes  achieve almost the upper theoretical bound of the embedding efficiency for  covers in the range of  1000 - 1500 bits, which is at least an order of magnitude smaller than the values reported in related works.

\vspace{.2cm}
\end{abstract}

\section{Introduction}
A widely accepted security model for steganographic systems was given in \cite{cachin1998information}. It is modelled as security in a form of visual and statistical undetectability. How much data can safely be embedded in a cover without being detected is given in \cite{harmsen2005capacity}.

One method for achieving undetectability is  matrix embedding. It has been first informally introduced in \cite{Crandall-Matrix-Embedding-1998} and more formally in \cite{vanDijk2001} and \cite{galand2003information}. It is a steganographic method that uses $(n, k)$ linear binary codes $\mathcal{C}$ to transmit messages of length $n-k$ bits, embedded into arbitrary covers of $n$ bits using as small as possible number of changes.
Matrix embedding addresses two important design goals in the steganography schemes: 
1. To achieve as high ratio as possible between the size of the embedded message and the size of the cover message (so called \emph{payload}); 2. To achieve as high security as possible in a form of visual and statistical undetectability.

The so called \emph{cover} in practice can come from different sources such as binary images, textual or binary files, line drawings, three-dimensional models, animation parameters, audio or video files, executable code, integrated circuits and many other sources of digital content \cite[Ch.1]{Cox:2007:DWS:1564551}.

In \cite{galand2003information} the theoretical bound  that is achievable with matrix embedding was given and it was shown that random linear codes asymptotically achieve the theoretical bound.

\subsection{Related Work}
Having just a theoretical result that long random linear codes can achieve the embedding capacity, is far from satisfactory in practice. Thus in a series of works, different practical matrix embedding algorithms have been proposed. 
In \cite{fridrich2006matrix}, two practical schemes for matrix embedding based on random linear codes and simplex codes were proposed. 
These schemes use relatively small values of $k\le 14$, one reason being that the embedding algorithm takes $O(n 2^k)$ operations.
Modifications of the schemes from \cite{fridrich2006matrix} with improved efficiency were proposed in a series of papers such as \cite{gao2009improving,chen2010efficient,wang2012fast,liu2014adaptive,mao2014fast}. These schemes, although efficient, do not offer embedding efficiency close to the upper theoretical bound.
In \cite{fridrich2007practical}, the authors propose to use low-density generator matrix (LDGM) codes (defined in  \cite{wainwright2005lossy}) for matrix embedding. Because of the efficient decoding algorithms for LDGM codes, the proposed schemes are quite fast and achieve very good embedding efficiency close to the theoretical bound, but for large $n$ in the range $10^5 - 10^6$.

\subsection{Our Contribution}
The contributions of this paper are severalfold: 1. We define a new family of binary linear codes with a generator matrix $G$ that has a specifically designed staircase random block structure. 2. We propose an efficient list decoding algorithm for these codes. 3. We perform an initial theoretical analysis of the stability and the complexity of the decoding algorithm. 4. We use these codes for matrix embedding. 5. We report the theoretical and experimental results in comparison with similar codes. The results show that our codes, while practical for matrix embedding, are also competitive with the best known codes. In particular, our codes achieve almost the upper theoretical bound of the embedding efficiency for the length of the cover in the range of  $1000 - 1500$ which is at least an order of magnitude smaller than the values reported in related works.

\section{Basics of Matrix Embeding}\label{sec:Prelim}
Throughout the paper, we will denote by $\mathcal{C}\subseteq \mathbb{F}^n_2$ a binary $(n, k)$ code of length $n$ and dimension $k$. We will denote the $k\times n$ generator matrix of the code by $\mathbf{G}$, and the $(n-k)\times n$ parity check matrix by $\mathbf{H}$.
The Hamming distance between $\mathbf{x}, \mathbf{y}\in \mathbb{F}^n_2$ will be denoted by $d(\mathbf{x},\mathbf{y})$, and the Hamming weight of a word $\mathbf{x}\in \mathbb{F}^n_2$  by $wt(\mathbf{x})$.

A crucial characteristic of a code that we will use is its covering radius $R$,  defined as $ R=\max_{\mathbf{x}\in\F_2^n} d(\mathbf{x},\mathcal{C})$
 where $ d(\mathbf{x},\mathcal{C})=\min_{\mathbf{c}\in\mathcal{C}}{d(\mathbf{x},\mathbf{c})}$ is the distance of $\mathbf{x}$ to the code $\mathcal{C}$. The average distance to a code $\mathcal{C}$ \cite{fridrich2006matrix} is defined by $2^{-n}\sum_{\mathbf{x}\in\F_2^n}d(\mathbf{x},\mathcal{C})$, and it represents the average distance between a randomly selected word from $\F_2^n$ and the code $\mathcal{C}$.  Clearly, $R_a \leq R$.

Suppose that we want to embed a message $\mathbf{m}\in \mathcal{M}=\mathbb{F}^{n-k}_2$ in a given cover object $\mathbf{y}\in \mathbb{F}^n_2$. Without loss of generality, both can be considered as random binary strings. Furthermore, the position of the cover object in a document (the block where the message is to be embedded) is known to both the sender and the recipient. 

\begin{definition}
A \emph{steganographic scheme} on $\F^n_2$ with a distortion bound $R$ is a pair of embedding and extraction mappings $Emb: \F^n_2 \times \mathcal{M} \rightarrow \F^n_2$,   $Ext: \F^n_2 \rightarrow  \mathcal{M}$, such that
\begin{eqnarray*}
Ext(Emb(\mathbf{y},\mathbf{m}))=\mathbf{m},& \forall \mathbf{y}\in \F^n_2, \forall\mathbf{m}\in \mathcal{M},\ and \\
d(\mathbf{y},Emb(\mathbf{y},\mathbf{m}))\leq R.&  
\end{eqnarray*}
\end{definition}

The embedding rate, or relative message length, is the value $\alpha = (n-k)/n$, and the lower embedding efficiency is the value $e=(n-k)/R$. For the average absolute distortion (average number of changes) $R_a$,  we have the average embedding efficiency $e_a=(n-k)/R_a$. 

 Let $\mathcal{C}$ be a binary $(n, k)$ code with a generator matrix $\mathbf{G}$ and a parity check matrix $\mathbf{H}$, both given in a systematic form. We assume that the sender and the recipient share the matrix $\mathbf{H}$ (and thus the matrix $\mathbf{G}$ as well). 
 
Algorithm~\ref{alg:matrixembedding}, taken from \cite{galand2003information} describes a general matrix embedding scheme for any linear code $\mathcal{C}$.   

\begin{figure}[h]
\begin{center}
\fbox{
  \parbox{.76\textwidth}{\smallskip
\textbf{Algorithm 1}: Matrix embedding
}}
\fbox{
  \parbox{.7535\textwidth}{\smallskip
$Emb(\mathbf{y},\mathbf{m})$:
\begin{enumerate}
\item Set $\mathbf{z}=\mathbf{y}\mathbf{H}^{\top}+\mathbf{m}$.

\item Let $\mathbf{c_0}$ be any such that $\mathbf{c_0}\mathbf{H}^{\top}=\mathbf{z}$.

\item Find the closest codeword $\mathbf{x}\mathbf{G}$ to $\mathbf{c_0}$ using some efficient algorithm.

\item Set $\mathbf{e}=\mathbf{c_0} + \mathbf{x}\mathbf{G}$.

\item Embed the message as $\mathbf{y'}=\mathbf{y}+\mathbf{e}$.
\end{enumerate}

\textbf{Output:} A stego object $\mathbf{y'}$.\smallskip\smallskip

$Ext(\mathbf{y'})$:\smallskip

\ \ Extract the message as $\mathbf{m}=\mathbf{y'}\mathbf{H}^{\top}$.\smallskip

\textbf{Output:} The extracted message $\mathbf{m}$.\smallskip
} }
\end{center}
\caption{Algorithm for matrix embedding using code $\mathcal{C}$}
\label{alg:matrixembedding}
\vspace{-.3cm}
\end{figure}

The crucial step in Algorithm 1 is Step 3, \ie finding the closest codeword $\mathbf{x}\mathbf{G}$ to $\mathbf{c_0}$. Thus, the performance of such a scheme is determined by the covering radius $R$ of the code $\mathcal{C}$ which guarantees a lower bound on the embedding efficiency, but also on the average distance to the code $R_a$ which coincides with the average distortion. 
Therefore, the art of designing a practical steganographic scheme lies in finding codes of small average distance $R_a$ for which  efficient algorithms for finding a codeword within $R_a$ exist.
However, both problems are known to be particularly difficult and challenging. 

It is known \cite{fridrich2006matrix} that when $n \rightarrow \infty$, random codes asymptotically achieve the upper bound for the embedding efficiency $e$ for a given embedding rate $\alpha$:
\begin{equation}\label{eq:bound}
e \leq \frac{\alpha}{H^{-1}(\alpha)}, \ 0\leq \alpha\leq 1
\end{equation}
where $H$ is the binary entropy function.

The authors of \cite{fridrich2007practical} report several LDGM codes of dimension $n=10000$ and $n=100000$ with extremely good embedding efficiency. However, there are no known (to the authors' knowledge) codes, of smaller  dimension, close to the bound~\eqref{eq:bound}, or to the codes from \cite{fridrich2007practical}. Furthermore, to the authors' knowledge, the codes reported in \cite{fridrich2007practical} have currently the best performance regarding embedding efficiency.

\section{Staircase-Generator Codes}
We consider a binary $(n, k)$ code $\mathcal{C}$ with the following generator matrix in standard form:
	\begin{equation}\label{eq:G-matrix}
		\input{G-matrix-diagram}
	\end{equation}			

Each $B_i$ is a binary matrix of dimension $k_i \times n_i$ whose structure will be discussed shortly, and each $B'_i$ is a random binary matrix of dimension $\sum_{j=1}^{i-1} k_j \times n_i$,
so that $k = k_1 + k_2 + \dotsb + k_v$ and $n = k + n_1 + n_2 + \dotsb + n_v$. Further we set  $K_i = k_1 + \dotsb + k_{i}$ and $N_i = n_1 + \dotsb + n_{i}$. We will call these codes - \emph{Staircase-Generator codes}.

Note that a code with generator matrix of the form \eqref{eq:G-matrix} can be considered as a generalization of at least two well known constructions, taking for example the extended direct sum (EDS) and the amalgamated direct sum (ADS) construction \cite{GrahamS85}. Indeed, it can be seen that the codes \eqref{eq:G-matrix} are a generalization of the EDS construction of the codes with generator matrices $[I_{K_1} B_1],\dots, [I_{K_v} B_v]$, since the matrices $B'_i$ are chosen at random. On the other hand, \eqref{eq:G-matrix} can be seen as an amalgamation of the codes 
with generator matrices 
\begin{equation*}
\left(\ \ I_{K_{i-1}\ \ }\begin{array}{ccccc}
B_1 & B'_2 & \dots & & B'_{i-1}\\
0	& B_2  & & \ & \\
	& 0  & & \ & \\
\vdots &\dots  &  & & \\
 0	& \  & & \!\!\!\!0  & B_{i-1} \\
\end{array}\right) {\rm \ \ and\ } 
\left( I_{K_{i}}\begin{array}{c}
 B'_{i}\\
 B_{i} \\
\end{array}\right)
\end{equation*} for each $i\in\{2,\dots,v\}$.

We should emphasize that, we do not impose any condition on the normality of the codes being amalgamated. In the standard ADS construction, such conditions are necessary in order to prove the improvement on the covering radius. However, in the general case where more than one coordinate is amalgamated, it is much more difficult to theoretically estimate the improvement, and some attempts have not given the desired results \cite{GrahamS85}. 

\subsection{An Algorithm for Matrix Embedding Using Staircase-Generator Codes}
We describe a general list decoding algorithm (Algorithm~\ref{alg:decoding}) for the code $\mathcal{C}$,
that can be used for finding a codeword close to a given random $\mathbf{c_0}$.
Under the condition that Algorithm~\ref{alg:decoding} is efficient, we immediately obtain an efficient variant of Algorithm~\ref{alg:matrixembedding} for matrix embedding. 
Thus, two important questions about Algorithm~\ref{alg:decoding} need to be answered: How efficient it is and what is the expected weight of the obtained $\mathbf{e}$.
In order to answer these questions, we discuss several different design choices.

First, we need to fix some notations. Let $\mathbf{G}_i$ denote the submatrix of the generator matrix $\mathbf{G}$ of size $K_i \times (K_i+N_i)$ 
consisting of an identity matrix $I_{K_i}$ concatenated with the matrix of the entries from the first $K_i$ rows  and the columns $k+1,\dots,k+n_i$ of $\mathbf{G}$. Further, let $w_b$ denote a small constant that we will refer to as \emph{round weight limit}.

\begin{figure}[h]
\begin{center}
\fbox{
  \parbox{.85\textwidth}{\smallskip
\textbf{Algorithm 2}: Decoding
}}
\fbox{
  \parbox{.85\textwidth}{\smallskip
\textbf{Input:} A vector $\mathbf{c_0} \in \mathbb{F}_2^n$, a generator matrix $\mathbf{G}$ of the form (\ref{eq:G-matrix}), and a starting weight limit $w_1$.\smallskip

\textbf{Output:} A vector $\mathbf{e} \in \mathbb{F}_2^n$ of small weight, such that $\mathbf{x}\mathbf{G} = \mathbf{e} + \mathbf{c_0}$, for some $\mathbf{x} \in \mathbb{F}_2^k$.\smallskip

\textbf{Procedure:}

Let $\mathbf{x}_i$ represent the first $K_i$ bits of the unknown vector $\mathbf{x}$.  
During decoding, we will maintain lists $L_1, L_2, \dotsc, L_v$ of triples $(\mathbf{x}_i,\mathbf{e}_i,\mathbf{b}_i)$ where $|\mathbf{x}_i|=K_i$, $|\mathbf{e}_i|=K_i+N_i$, $|\mathbf{b}_i|=K_i+N_i$, that satisfy
\begin{equation}\label{eq:testweight}
\mathbf{x}_i\mathbf{G}_i = \mathbf{e}_i + \mathbf{b}_{i-1}, {\rm\ and\ } wt(\mathbf{e}_i)\leq w_i
\end{equation}

\emph{Step~0}: Let $\mathbf{b}_0=\mathbf{c}_0[1\dots k_1]||\mathbf{c}_0[k+1 \dots k+n_1]$ 
and $\mathbf{x}_0$ and $\mathbf{e}_0$ be vectors of dimension $0$. Set a starting list
 $L_0=[(\mathbf{x}_0,\mathbf{e}_0, \mathbf{b}_0)]$.
\smallskip

\emph{Step~$1 \leq i \leq v$}: 

For each $(\mathbf{x}_{i-1},\mathbf{e}_{i-1},\mathbf{b}_{i-1}) \in L_{i-1}$, add all $(\mathbf{x}_{i},\mathbf{e}_{i},\mathbf{b}_{i})$ to $L_i$, that satisfy \eqref{eq:testweight} and $ wt(e_i)\leqslant w_b$, 
where $\mathbf{x}_{i}=\mathbf{x}_{i-1}||x_i$, $\mathbf{e}_{i}=\mathbf{e}_{i-1}||e_i$ and $x_i$ and $e_i$ are unknown parts of $\mathbf{x}_{i}$, $\mathbf{e}_{i}$. Further, set
$$\mathbf{b}_{i}=\mathbf{b}_{i-1}||\mathbf{x}_i||\mathbf{x}_iB'_i.$$

 If $|L_i|< L$ 
 then $w_{i+1}=w_i+1$, otherwise $w_{i+1}=w_i$. 

\smallskip
			
\textbf{Return: } $(\mathbf{x}_{v},\mathbf{e}_{v},\mathbf{b}_{v})\in L_v$ with minimal $wt(\mathbf{e}_{v})$.
}}
\end{center}
\caption{Algorithm for finding a codeword close to a given random word}
\label{alg:decoding}
\vspace{-.5cm}
\end{figure}

\subsection{Choosing the Matrices $B_i$}\label{subsec:MatricesBi}
Note first, that at each step \emph{Step~$i$}, effectively, we work with the codes $\mathcal{C}_i$ with generator matrices $G_i=[I_{k_i}\ B_i]$. In particular, we find all the codewords $x_i G_i$ such that  $\mathbf{x}_i\mathbf{G}_i$ is within a radius $w_i$ of a word $\mathbf{b}_{i-1}$. 
Thus, it is desirable that the codes $\mathcal{C}_i$ have as small as possible covering radius. For efficiency reasons of the algorithm, these codes also need to be relatively small. 
Luckily, for small codes, it is not hard to find exactly the least possible covering radius. The authors of  \cite{BaichevaB10} provide a nice classification of the least covering radius of small codes. 
For our purposes, for  $i> 1$, we use small codes of covering radius $R=1$, and some of the possible choices are given in Table~\ref{tab:covRad}. For $i=1$, we typically, choose the code $C_1$ to be of very high rate. Thus we don't need to choose it with particular properties (although it is possible), since it is expected that it has low covering radius. 

\begin{table}[H]
\caption{Concrete codes $(n,k)R$ of covering radius $R=1$, with generator matrix $[I\ B]$}
\label{tab:covRad}
\centering
\vspace{.2cm}
\begin{tabular}{|c|c|}
\hline
 \parbox{2cm}{\begin{center}
 $(n,k)R$
 \end{center}} & \parbox{2cm}{\begin{center}
 matrix $B$
 \end{center}}\\
  \hline
  $(2,1)1$ & 
 $
\left(\begin{array}{c}
 1 \\
\end{array}\right)$\\
  \hline
  $(3,1)1$ & 
 $
\left(\begin{array}{cc}
 1& 1 \\
\end{array}\right)$\\
  \hline
  $(3,2)1$ & 
 $
\left(\begin{array}{c}
 1 \\
 1 \\
\end{array}\right)$\\
\hline
 $(4,3)1$ & 
 $
\left(\begin{array}{c}
 1  \\
 1  \\
 1  \\
\end{array}\right)$\\
\hline

\end{tabular}
\begin{tabular}{c}
\  \\
\end{tabular}
\begin{tabular}{|c|c|}
\hline
 \parbox{2cm}{\begin{center}
 $(n,k)R$
 \end{center}} & \parbox{2cm}{\begin{center}
 matrix $B$
 \end{center}}\\
  \hline
 $(5,3)1$ & 
 $
\left(\begin{array}{cc}
1 & 1  \\
1 & 1  \\
1 & 1  \\
\end{array}\right)$\\
 \hline
  $(5,4)1$ & 
 $
\left(\begin{array}{c}

 1 \\
  1 \\
   1 \\
  1 \\
\end{array}\right)$\\
\hline
\end{tabular}
\end{table}  
\captionsetup{font={footnotesize,rm},justification=centering,labelsep=period}%

\subsection{Average Distortion Estimates}\label{subsec:bounds}
In order to keep the complexity of Algorithm~\ref{alg:decoding} low, not only need the codes $G_i$ be small, but also the size of the maintained lists needs to be low as well. Our algorithm chooses the size of the lists to be bounded by some constant $L$. In this case we can quite accurately estimate the average distortion  achievable by Algorithm~\ref{alg:decoding}, that we will denote by $R_{alg}$.
Note that in general $R_a \leqslant R_{alg}$.
The following theorem provides an estimate of  $R_{alg}$. 

\begin{theorem}\label{thm:Ra}
Let $\mathcal{C}$ be a $(n,k)$ code with a generator matrix $\mathbf{G}$ of the form (\ref{eq:G-matrix}).  Then, in Algorithm~\ref{alg:decoding}, with  round weight limit $w_b$, we have that at each step $i$, $1 \leqslant i \leqslant v$,

1. The expected number of vectors $\mathbf{e}_i\in\F^{K_i+N_i}_2$ such that $wt(\mathbf{e}_i)=j$, is $V_i(j)$ where:
\begin{eqnarray*}
V_1(j)&\!\!=\!\! &\frac{\binom{k_1+n_1}{j}}{2^{n_1}}, {\ \rm for\ } 0 \leqslant j \leqslant w_1,\\
V_i(0)&\!\!=\!\! &\frac{1}{2^{n_i}} V_{i-1}(0),\ \ V_i(j)\ =\ 0 {\rm\ for\ } j > w_i {\rm\ and}\\ 
V_i(j)&=&\sum_{\ell=0}^{\min(j,w_b)}{\frac{\binom{k_i+n_i}{\ell}}{2^{n_i}}\cdot V_{i-1}(j-\ell)}  {\rm\ for} \ \  0 < j \leqslant w_i.
\end{eqnarray*}
\vspace{-.3cm}

2. The expected size of the list $L_i$ is $\displaystyle\sum_{j=0}^{w_i}{V_i(j)}$.

3. $[R_{alg}]$ is the smallest $j$ such that $V_{v}(j) \geqslant 1$.
\end{theorem}
\begin{proof}
1. Since there are $\binom{k_1+n_1}{j}$ different ways to make exactly $j$ changes in $k_1+n_1$ bits, and there are $2^{k_1}$ codewords in $\mathcal{C}_1$ (the code with generator matrix $[I\ B_1]$), it can be expected that among all codewords,  $2^{k_1}\frac{\binom{k_1+n_1}{j}}{2^{k_1 + n_1}} = V_1(j)$ will be at distance exactly $j$ from a random word in $\F^{k_1+n_1}_2$.

Further, if at Step $i$, $V_{i-1}(0)$ is the expected number of vectors $\mathbf{e}_{i-1}$ of weight $0$, in the next step, this number reduces to $2^{k_i}\frac{1}{2^{k_i + n_i}}V_{i-1}(0)$. 

The value of  $V_i(j)$ can be obtained as follows. At each Step $i$ we test each $e_i$ of weight $wt(e_i)=\ell\leqslant w_b$ and each $\mathbf{e}_{i-1}$ in the list $L_{i-1}$ for consistency with~\eqref{eq:testweight}. If consistent, $\mathbf{e}_i$ belongs to $V_i(j)$ only if $wt(\mathbf{e}_{i-1}||e_i)=j$, \ie only if  $\mathbf{e}_{i-1}$ belongs to $V_{i-1}(j-\ell)$. Again, on average, $1/2^{n_i}$ of the tested $e_i$ satisfy~\eqref{eq:testweight}. From here, we immediately obtain the claimed value for $V_i(j)$.

2. Since the list $L_i$ contains all vectors $\mathbf{e}_{i}$ of weight $\leqslant w_i$, we can estimate its size by taking the sum of all $V_i(j)$, $0 \leqslant j \leqslant w_i$.

3. In the last Step $v$, if $V_v(j) \geqslant 1$, then we can expect that in the list $L_v$, there is an element  $\mathbf{e}'_{v}$ of weight $j$. Taking the smallest such $j$ determines the expected average distortion.\qed
\end{proof}

As mentioned in Subsection~\ref{subsec:MatricesBi}, for better results, in practice, we use matrices $B_i$ as in Table~\ref{tab:covRad}, that guarantee that the codes $\mathcal{C}_i$, $i > 1$ with  generator matrix $[I_{k_i}\ B_i]$ have covering radius 1. Thus we can restrict the choice of the round weight limit $w_b$ to small values, typically, $1$ or $2$. Calculating the average distortion for concrete parameters using Theorem~\ref{thm:Ra} shows that the choice of $w_b=2$ gives better results (cf. Section~\ref{sec:practical}). 

\subsection{List Estimates}
Theorem~\ref{thm:Ra} not only provides  the average distortion obtained using Algorithm~\ref{alg:decoding}, it also gives  the size of the lists $L_i$ at each step of Algorithm~\ref{alg:decoding}. For the efficiency of the algorithm it is important to know at what conditions the lists grow or decrease from one step to another.

Let $B_i(j)$ denote the ball of radius $j$ around the zero vector of length $K_i+N_i$, containing the vectors $\mathbf{e}_i$, with $wt(\mathbf{e}_i)\leqslant j$ obtained in Algorithm~\ref{alg:decoding}. Further, we denote by $|B_i(j)|$ the expected size of $B_i(j)$. Then, using the notation from Theorem~\ref{thm:Ra}, $|B_i(j)|=\sum_{s=0}^{j}{V_i(s)}$, and  $|L_i|=|B_i(w_i)|$. Directly from Theorem~\ref{thm:Ra}, we have the following lemma.

\begin{lemma}\label{lem:lists}
Let $w_b=2$. Then, 
\begin{eqnarray}
&& \hspace{-.7cm}{\rm when\ } w_{i+1}=w_{i}\ \nonumber\\
&& |L_{i+1}|=|B_{i+1}(w_{i+1})|=
\frac{1}{2^{n_{i+1}}}|L_i|+ \frac{(k_{i+1}+n_{i+1})}{2^{n_{i+1}}}|B_{i}(w_i\!-\!1)|\!+\!\frac{\binom{k_{i+1}+n_{i+1}}{2}}{2^{n_{i+1}}}|B_{i}(w_i\!-\!2)|, \quad\ \ \\ 
&&\hspace{-.7cm}{\rm and\ when}\ w_{i+1}=w_{i}+1,\ \nonumber\\
&&|L_{i+1}|=|B_{i+1}(w_{i+1})|= \frac{(k_{i+1}\!+\!n_{i+1}+1)}{2^{n_{i+1}}}|L_{i}|+\frac{\binom{k_{i+1}+n_{i+1}}{2}}{2^{n_{i+1}}}|B_{i}(w_i\!-\!1)|.\quad\quad
\end{eqnarray} \qed
\end{lemma}

Before continuing, for clarity of the exposition, we make two simplifications.

First, without loss of generality, we can assume that at a given Step $i$ only $|B_i(w_i)|, |B_i(w_i-1)|,\dots,|B_i(w_i-t)| > \frac{1}{2^s}$,  while $|B_i(0)|,\dots,|B_i(w_i-t-1)|\leqslant \frac{1}{2^s}$ for some integer $s$. Thus for big enough $s$, we will take $|B_i(0)|=0,\dots,|B_i(w_i-t-1)|=0$. (Indeed, this can be safely  assumed for $|B_i(0)|$ since $|B_i(0)|\rightarrow 0$ as $i$ grows.).
Second, we assume that $k_i=k_2$, $n_i=n_2$ for all $i\geqslant 2$. In practice, we will typically choose such parameters.

\begin{proposition}\label{lem:Balls}
Let $w_b = 2$. Let at Step $i$, $|B_i(j)|>0$, for $j \geqslant w_i-t$, and $|B_i(j)|=0$, for $j < w_i-t$. Then, if $w_{i+1}=w_i$, we have that:
\begin{eqnarray}
|B_i(w_i-t)| & > & |B_{i+1}(w_i-t)|,\label{eq:steadystate1}\\ 
|B_i(w_i-t+1)| & > & |B_{i+1}(w_i-t+1)|\ {\rm if\ and\ only\ if}\ \frac{k_2+n_2}{2^{n_2}-1}|B_i(w_i-t)| < |B_i(w_i-t+1)|,\label{eq:steadystate2}\\ 
|B_i(j)| & > & |B_{i+1}(j)|,  w_i-t+1 < j \leqslant w_i, {\rm \ if\ and\ only\ if}\ \nonumber\\
& &  \frac{(k_2+n_2)|B_i(j\!-\!1)|+\binom{k_2+n_2}{2}|B_i(j\!-\!2)|}{2^{n_2}-1} < |B_i(j)|.\label{eq:steadystate3}
\end{eqnarray}
\end{proposition}
\begin{proof}
The proof is rather straightforward. We have:
\begin{eqnarray*}
 |B_{i+1}|(w_i-t) = \frac{1}{2^{n_2}} |B_i(w_i-t)| < |B_i(w_i-t)|,
\end{eqnarray*}
since $ |B_i(w_i-t-1)|=|B_i(w_i-t-2)|=0$ from the condition.
Further,
\begin{eqnarray*}
 |B_{i+1}(w_i\!-\!t\!+\!1)|=\frac{|B_i(w_i\!-\!t\!+\!1)|+ (k_2+n_2)|B_i(w_i\!-\!t)|}{2^{n_2}}
\end{eqnarray*} 
Now, directly, it can be seen that $\frac{k_2+n_2}{2^{n_2}-1}|B_i(w_i-t)| < |B_i(w_i-t+1)|$ is equivalent to $ |B_{i+1}(w_i-t+1)|< |B_i(w_i-t+1)|$.

Very similarly, the last claim follows directly from the expression for $|B_{i+1}(j)|$.\qed
\end{proof}

The previous simple proposition implies that if at Step $i$, 
\eqref{eq:steadystate2} and \eqref{eq:steadystate3} hold, 
in the next Step $i+1$, for all $j$,  $|B_{i+1}(j)| < |B_{i}(j)|$. 
When $j=w_i$, this guarantees that at Step $i+1$ the size of the list $L_{i+1}$ will be smaller than in the previous step. However, even if \eqref{eq:steadystate2} or \eqref{eq:steadystate3} for some weight $j$ is not satisfied at Step $i$, at a later step $i+s$, if $w_{i+s}=w_i$, an equivalent expressions to \eqref{eq:steadystate3} (for $i+s$ instead of $i$) will hold. 

\begin{proposition}\label{lem:listdecrease}
Let $w_b=2$. There exists an integer $s$, s.t. if $w_i=w_{i+1}=\dots=w_{i+s}$ then $|L_{i+s+1}| < |L_{i+s}|$. 
\end{proposition}
\begin{proof}
We will show that \eqref{eq:steadystate3} will eventually be satisfied. From here, the claim will follow immediately. 

Suppose $|B_i(w_i\!-\!t\!+\!1)|=p\frac{k_2+n_2}{2^{n_2}-1}|B_i(w_i\!-\! t)|$, where $0<p<1$.
Then $|B_{i+1}(w_i-t)|=\frac{1}{2^{n_2}}|B_{i}(w_i-t)|$, and it can easily be verified that $|B_i(w_i\!-\!t\!+\!1)|=\frac{k_2+n_2}{2^{n_2}-1}|B_i(w_i\!-\! t)(p-1+2^{n_2})|$, and since $(p-1+2^{n_2})> 1$ for every $n_2\geqslant 1$, we get that $|B_{i+1}(w_i\!-\!t\!+\!1)|>\frac{k_2+n_2}{2^{n_2}-1}|B_{i+1}(w_i\!-\! t)|$.\\
Next, suppose, \eqref{eq:steadystate2} and \eqref{eq:steadystate3} hold for $j<s$, but $|B_i(s)|=p_1\frac{(k_2+n_2)|B_i(s-1)|+\binom{k_2+n_2}{2}|B_i(s-2)|}{2^{n_2}-1}$, where $0<p_1<1$.
Similarly, let $|B_{i+1}(s)|=p_2\frac{(k_2+n_2)|B_{i+1}(s-1)|+\binom{k_2+n_2}{2}|B_{i+1}(s-2)|}{2^{n_2}-1}$, for some $p_2$.
Then, it can be verified that $|B_{i+1}(s)|=\frac{p_1+2^{n_2}-1}{2^{n_2}}\cdot\frac{(k_2+n_2)|B_i(s-1)|+\binom{k_2+n_2}{2}|B_i(s-2)|}{2^{n_2}-1}$. From the assumption that $|B_i(s-1)|>|B_{i+1}(s-1)|$, $|B_i(s-2)|>|B_{i+1}(s-2)|$, we get that 
$$p_2> \frac{p_1+2^{n_2}-1}{2^{n_2}},\ \ie\ 1-p_1> 2^{n_2}(1-p_2).$$
This means that either $p_2>1$, in which case either \eqref{eq:steadystate3} is satisfied for step $i+1$, or if $p_2\leqslant 1$, then it is $2^{n_2}$ times closer to $1$ than $p_1$. Thus if we repeat the process for $i+2,i+3,...$, the sequence $p_1,p_2,p_3,\dots$ will either surpass $1$ or will exponentially fast approach $1$. When $p_\ell<1$ is very close to $1$, (depending on the parameters) we can expect that \eqref{eq:steadystate3} will be satisfied in the next step.\qed
\end{proof}

 For the practical parameters that will be given in Section~\ref{sec:practical}, it was observed that $|L_{i+s+1}| < |L_{i+s}|$ becomes true in just a couple of steps.

The previous discussion was concerned with Steps $i$, when the weight $w_i$, does not change throughout the steps. Since in this case, the size of the lists $|L_i|$ will eventually (after a few steps) start decreasing, at one point the condition $|L_i|< L$, will be satisfied, and  Algorithm~\ref{alg:decoding} will increase the weight to $w_{i+1}=w_i$. The goal of this step, is to increase the list again. 
A direct application of Lemma~\ref{lem:lists}, yields:
\begin{proposition}\label{{lem:listincrease}}
Let $w_b=2$. Then, when $w_{i+1}=w_{i}+1$, 
\begin{eqnarray}
|L_{i+1}|>|L_{i}| {\rm \ if\ and\ only\ if\ }
|B_{i}|> \frac{2^{n_2}\!-\!k_2\!-\!n_2\!-\!1}{\binom{k_2+n_2}{2}}|L_{i}|.\!\! \label{eq:growlist}
\end{eqnarray}\qed
\end{proposition}
For the chosen practical parameters from Section~\ref{sec:practical}, the right hand side in \eqref{eq:growlist} will always be negative, and thus increasing the weight $w_i$ will always increase the size of the list.

\subsection{Complexity of Algorithm~\ref{alg:decoding}}
We will first take a look at the first round since it is different from the rest. Filling up the list $L_1$ would require testing all $k_1+n_1$ bit words $\mathbf{e}_1$ of weight at most $w_i$ for consistency with \eqref{eq:testweight}, and finding the appropriate vectors $\mathbf{b}_1$. The complexity of this part is $\mathcal{O}((k_1+n_1)^{w_1}k_1n_2)$. This immediately implies that $w_1$ needs to be chosen very small, typically $\leqslant 3$.

Next, at each Step $i$, we consider only $e_i$ of weight $wt(e_i)\leqslant 2$. Testing all the elements from $L_{i-1}$, would thus take $|L_{i-1}|(k_2+n_2)^2$, and forming the list $L_i$ would take additional $|L_i|\mathcal{O}((k_1+(i-1)k_2)n_2)=|L_i|\mathcal{O}(k)$. Since we restrict the size of the lists to $\mathcal{O}(L)$, we have that the complexity of one round is  $\mathcal{O}(Lk)$. In total there are $v=\mathcal{O}(n)$ rounds, so the total complexity of the algorithm amounts to $\mathcal{O}(Lkn)$. 
Note that this is only a rough estimate of the complexity, without taking into account any implementation optimizations. 

\section{Practical Parameters and Experimental Results}\label{sec:practical}
Following the design choices justified in the previous sections, we have formed several codes suitable for practical use. For all codes, $w_1=2$ and $w_b=2$. The other parameters  are summarized in Table~\ref{tab:paramsCodes}.
\begin{table}[H]
\caption{Parameters for St-Gen codes for matrix embedding}
\label{tab:paramsCodes}
\vspace{.3cm}
\centering
\begin{tabular}{|c|c|c|c|c|c|c|}
\hline
\ \ $n_1$\ \  & \ \ $k_1$\ \ & \ \ $n_2$\ \ & \ \ $k_2$\ \ &  \quad \ $\mathcal{C}_i$\quad \ &  \ \ $\mathcal{C}$ ($n\approx 1000$) & \ \ $\mathcal{C}$ ($n\approx 1500$)\\
  \hline
   $2$ &  $14$ &  $2$ & $1$ &  $(3,1)1$ &  $(1000,343)$ &  $(1501,509)$  \\
   $1$ &  $14$ &  $1$ & $1$ &  $(2,1)1$ &  $(1001,507)$ &  $(1501,757)$ \\  
   $2$ &  $14$ &  $2$ & $3$ &  $(5,3)1$ &  $(999,603)$ &  $(1499,903)$ \\
   $1$ &  $14$ &  $1$ & $2$ &  $(3,2)1$ &  $(1002,672)$ & $(1500,1004)$ \\
   $1$ &  $14$ &  $1$ & $3$ &  $(4,3)1$ &  $(1003,755)$  &  $(1503,1130)$ \\
   $1$ &  $14$ &  $1$ & $4$ &  $(5,4)1$ &  $(1000,802)$  &  $(1500,1202)$ \\
\hline
\end{tabular}
\end{table}
\captionsetup{font={footnotesize,rm},justification=centering,labelsep=period}%

We have performed an extensive set of experiments to test the performance  of our codes in terms of their embedding efficiency $e(\alpha)$. The experiments were performed using an initial implementation in Magma \cite{Magma}. All the results presented in this section are obtained as an average over 50 experiments for each code. 
 
Figure~\ref{fig:resultsAll} provides a comparison of the St-Gen codes of length $n\approx 1500$ to other previously known codes: the binary Hamming codes \cite{BierbrauerF08}, the BDS codes \cite{BierbrauerF08}, the LDGM codes \cite{fridrich2007practical}, as well as to two theoretically derived stego-code families (SCF) of Golay codes and LDGM codes \cite{ZhangZW10}. 
It can be seen from the figure that our codes have approximately the same embedding efficiency as the LDGM codes, with one important difference: The presented performance of the LDGM codes is achieved for much bigger lengths, namely for $n=10000$.  
\begin{figure}[H]\centering
{\footnotesize$e(\alpha)$}\qquad\qquad\ \qquad\qquad\qquad\qquad\qquad\qquad\qquad\qquad\qquad\qquad\qquad\qquad\ \ \\
\hspace{-.0cm}\includegraphics[scale=.52]{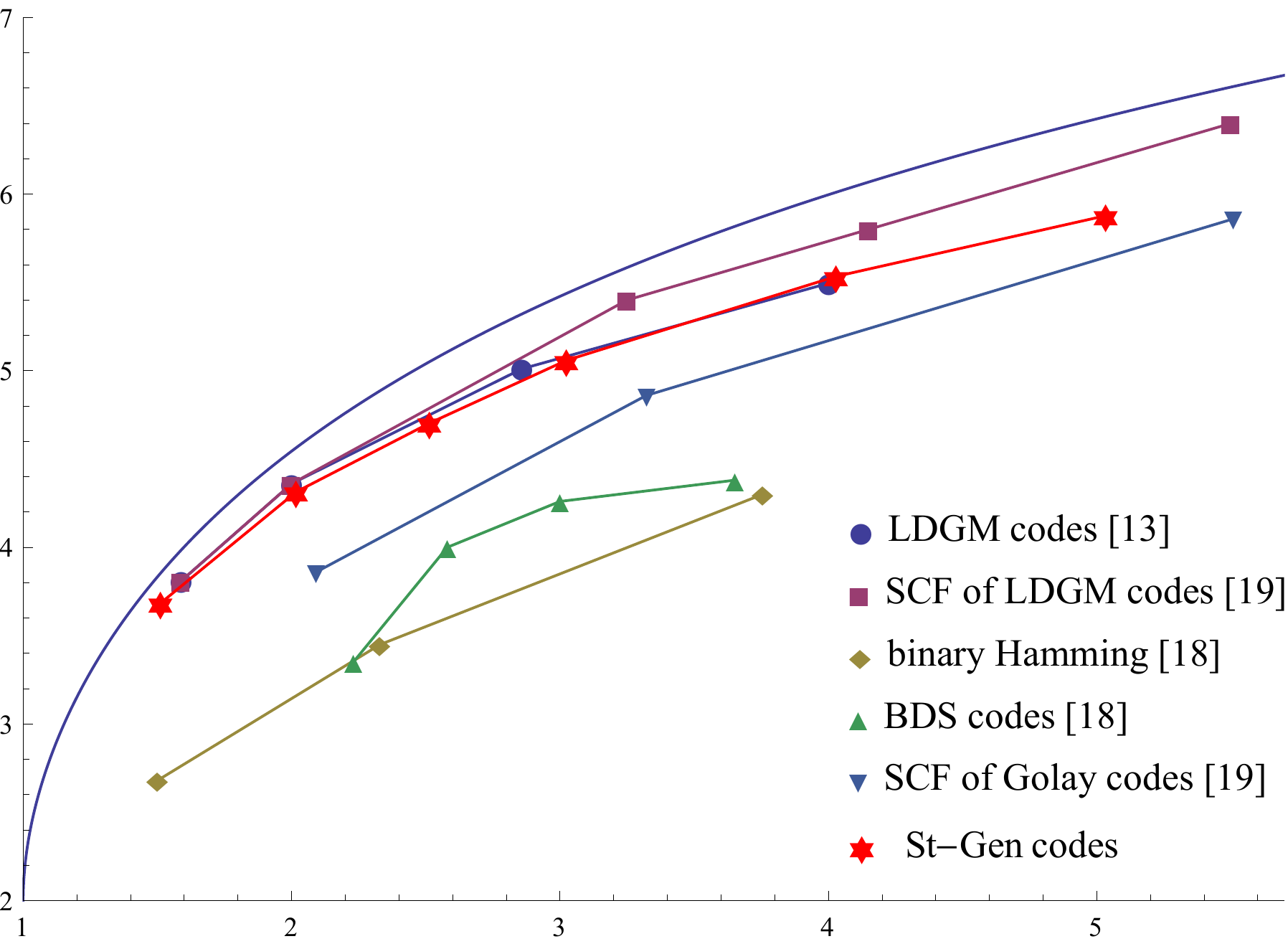}
{\footnotesize\hspace{-.5cm}$1/\alpha$} 
\caption{Comparison of the embedding efficiency of St-Gen codes  with other known codes for matrix embedding}  
\label{fig:resultsAll}
\vspace{-.4cm}
\end{figure}

For lengths $1000-1500$, the authors of \cite{fridrich2007practical} provide information only for $\alpha=1/2$.
Plotted together with our codes in Figure~\ref{fig:resultsDiffLengths}, we see that the embedding efficiency of the LDGM codes is quite smaller than that of the St-Gen codes. Our assumption is that the behaviour is similar for other embedding rates as well.

\begin{figure}[H]\centering
{\footnotesize$e(\alpha)$}\qquad\qquad\qquad\qquad\qquad\qquad\qquad\qquad\qquad\qquad\qquad\qquad\qquad\qquad\ \\
\includegraphics[scale=.52]{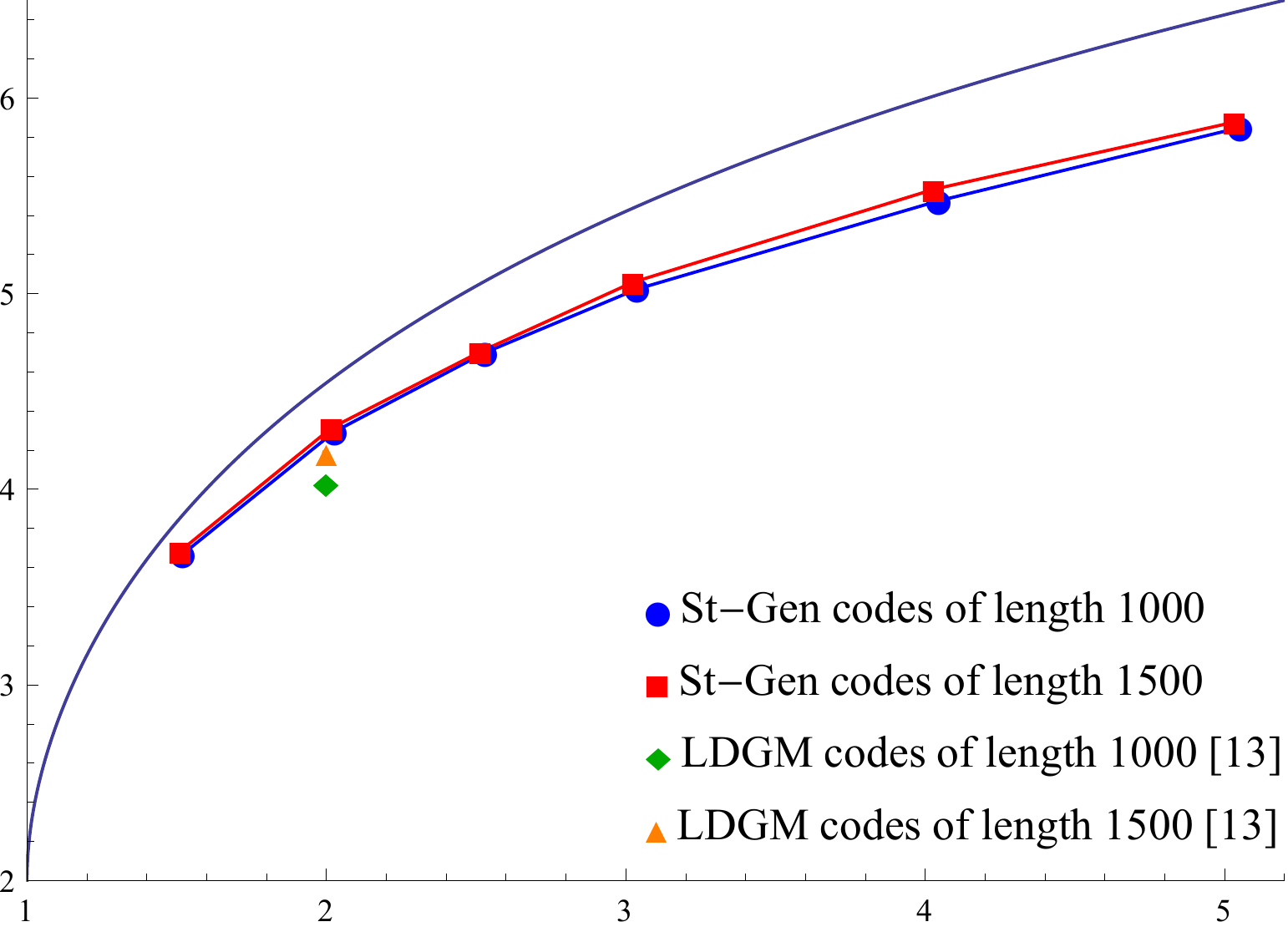}
{\footnotesize\hspace{-.0cm}$1/\alpha$}
\caption{Comparison of the embedding efficiency of St-Gen codes for different lengths of the codes.}
\label{fig:resultsDiffLengths}
\end{figure}
We have also tested how well the experimental results fit the theoretical estimates from Theorem~\ref{thm:Ra}. Figure~\ref{fig:resultsTheoryvsPractice} shows a small offset between the two, and it remains an open problem to investigate whether the offset has a true significance. The last Figure~\ref{fig:resultsDiffWeights}, shows the difference in performance, obtained experimentally, depending on the choice of the round weight limit $w_b$. 
\begin{multicols}{2}
\begin{figure}[H]
{\footnotesize$e(\alpha)$}\\
\includegraphics[scale=.5]{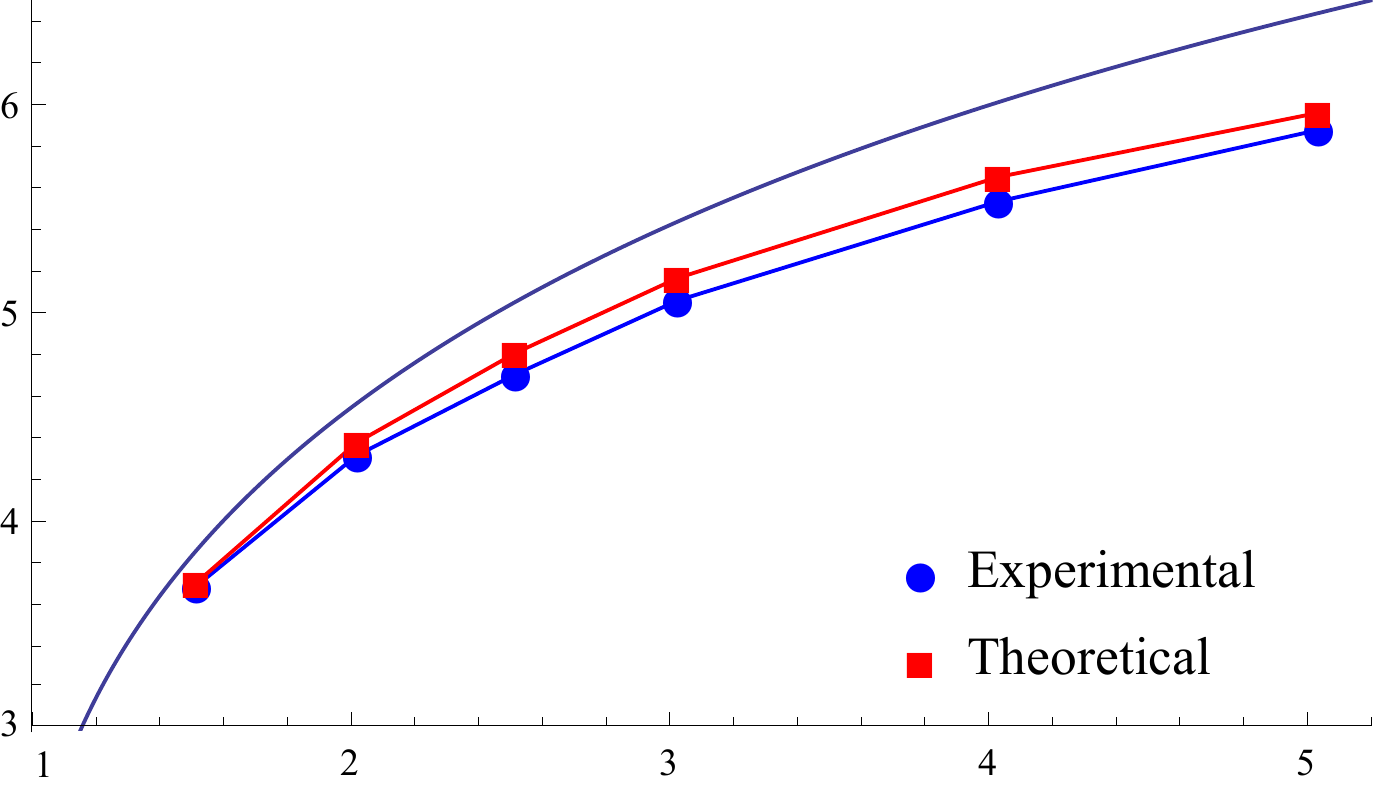} 
{\footnotesize\hspace{-.2cm}$1/\alpha$}
\caption{Comparison of the theoretical results from Theorem~\ref{thm:Ra} and the obtained practical results for St-Gen codes}
\label{fig:resultsTheoryvsPractice}
\vspace{-.4cm}
\end{figure}
\begin{figure}[H]
{\footnotesize$e(\alpha)$}\\
{\centering
\includegraphics[scale=.5]{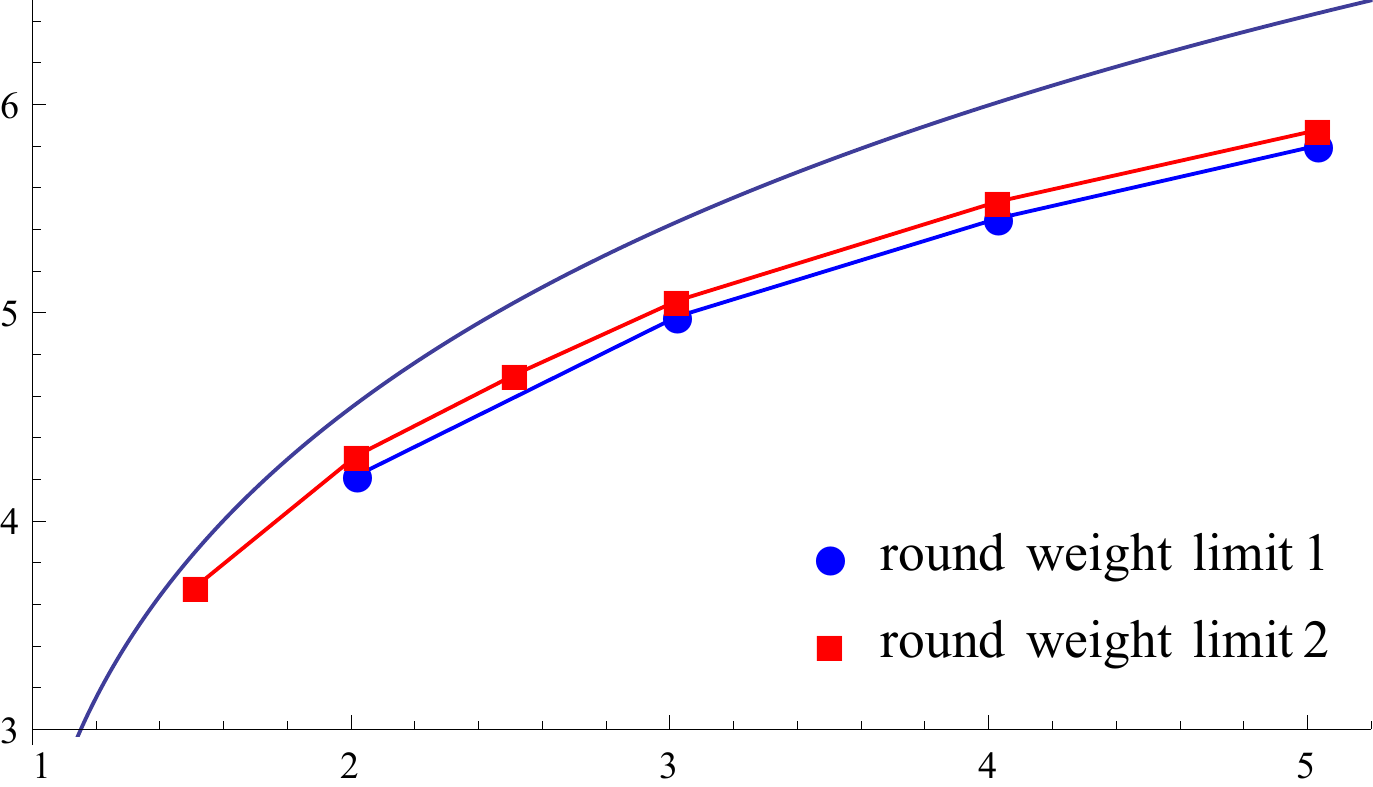} 
{\footnotesize\hspace{-.2cm}$1/\alpha$}}
\caption{Comparison of the embedding efficiency of St-Gen codes for different choice of the round weight limit in Algoritam~\ref{alg:decoding}.}
\label{fig:resultsDiffWeights}
\vspace{-.4cm}
\end{figure}
\end{multicols}

\section{Conclusions}
We introduced a new family of binary linear codes called Staircase-Generated codes, whose structure is suitable for applying list decoding techniques. We proposed an algorithm 
that searches for a close codeword, and as our theoretical analysis and experiments confirm, it shows very good performance results: Our codes achieve approximately the same embedding efficiency as the currently best codes for matrix embedding, but for lengths at least an order of magnitude smaller.

Having a proof of concept, our future work will be, firstly, directed towards making a fast, optimized implementation, for ex. in C, and applying different techniques for reducing the running time of the algorithm. 
On the theoretical side, we plan to extend the analysis to estimating the covering radius and the average distance of the codes for different and more general scenarios.

\section*{Acknowledgements}\small
We thank the anonymous referees for their comments which helped improve this work. The first author of the paper has been partially supported by FCSE, UKIM, Macedonia and the COINS Research School of Computer and Information Security, Norway.

\bibliographystyle{spmpsci}

\end{document}

%% file: G-matrix-diagram.tex
 \scalebox{.6}{\begin{tikzpicture}[
 		gray box/.style={top color=lightgray!30,bottom color=lightgray!30,middle color=lightgray!30},
 		B style/.style={gray}
 		]	  
 			 
	 \node[matrix of math nodes,matrix anchor=center, baseline=(M.center),
	 		left delimiter=(, 
	 		right delimiter=),
	 		ampersand replacement=\&, 
	 		inner sep=0pt,
   		row sep=4pt,
  		column sep=0pt, 
  		minimum size=17pt, 
   		nodes in empty cells,
   		column 3/.style ={inner xsep=10pt},
   		column 4/.style ={inner xsep=-10pt},
   		column 5/.style ={inner xsep=10pt},
   		column 6/.style ={inner xsep=10pt},
   		column 10/.style ={inner xsep=1pt}
   		] (M)
    {
      \phantom{B_1}\hspace*{60pt}\&\& \&	\&\phantom{B_2}	\&\phantom{B_3}	\& \phantom{B_w}\&\phantom{B_w}\&\phantom{B_w}\&\phantom{B_w}\&	\& \phantom{B_w}\\
     \&\& 	\&	  \&\phantom{B_2}	\&\phantom{B_3}	\&\&						 \& 						\&\phantom{B_w}	\& 	\&\phantom{B_w}				 \\
	\&\&\& 		\&\phantom{B_2}	\&\phantom{B_3}	\&\phantom{B_w}	\& \dots\& \&\phantom{B_w}	\&	\&\phantom{B_w}			 	\\
										\&\&							\&		\&\phantom{B_2} 							\&\phantom{0 \quad 0}\&\phantom{B_w}	\& 						 \&		 					\&							\&	\&			 	\\
										\&\&							\&		\& 							\&							\&							\&\dots\&							\&\phantom{B_w}	\&	\&			 	 \\
			  						\&\&\phantom{B_2}	\&		\&	\&							\&							\&\phantom{B_w}\&\phantom{B_w}\&\phantom{B_w}					\& 					\& \phantom{B_w} \\
    };
		
		\draw[shorten <= 5pt,shorten >= 5pt] (M-1-1.north west) -- node[scale=1.8, fill=white]{$I_k$} (M-6-4.south);		
		
		\draw[] ([xshift=12pt]M-1-4.north west) -- node[pos=0.15, fill=white, yscale=3]{} ([xshift=12pt]M-6-4.south west);

   	\node[left=10pt of M, scale=1.8] (G) {$\mathbf{G} =$};

   	\node[scale=2.5] at (M-5-6) {$0$};

   	\draw[] (M-2-5.north west) -- (M-1-5.north west) -- (M-1-7.north east);	
   	\draw[] (M-1-8.north east) -- (M-1-12.north east) -- (M-6-12.south east);
  	\draw[] (M-6-12.south east) -- (M-6-10.south west) -- (M-6-10.north west);
   	\draw[] (M-2-5.north west) -- (M-2-5.south west) -- (M-2-5.south east);
	\draw[] (M-2-5.south east) -- (M-3-5.south east) -- (M-3-7.south west);

		\begin{pgfonlayer}{background}  
			\path [fill=boxgray] (M-1-5.north west) rectangle node[scale=1.2,black]{$B_1$} (M-2-5.south east);  
			\path [fill=boxgray] (M-1-5.north east) rectangle node[scale=1.2,gray]{$B'_2$} (M-2-7.south west); 
			\path [fill=boxgray] (M-2-5.south east) rectangle node[scale=1.2,black]{$B_2$} (M-3-7.south west);  
			\path [decoration={random steps,segment length=10pt}, fill=boxgray] (M-1-7.north west) -- (M-4-7.west) decorate {-- (M-4-7.center) -- (M-1-7.north east)} -- cycle;
			\path [decoration={random steps,segment length=10pt}, fill=boxgray] (M-1-9.north west) decorate{-- (M-6-9.north)} -- (M-6-9.north east) -- (M-1-9.north east) -- cycle;
			\path [fill=boxgray] (M-1-10.north west) rectangle node[scale=1.2,gray]{$B'_v$} (M-2-12.south east);
			\path [fill=boxgray] (M-2-10.south west) rectangle (M-5-12.south east);
			\path [fill=boxgray] (M-5-9.south east) rectangle node[scale=1.2,black]{$B_v$} (M-6-12.south east);			
					
			\draw[densely dashed, draw=lightgray] (M-2-5.south east) -- (M-1-5.north east);
			\draw[densely dashed, draw=lightgray] (M-2-5.south east) -- (M-2-6.south east);%
			
			\draw[densely dashed, draw=lightgray] (M-3-7.south west) -- (M-1-7.north west);
			
			\draw[densely dashed, draw=lightgray] (M-6-10.north west) -- (M-1-10.north west);					\draw[densely dashed, draw=lightgray] (M-6-10.north west) -- (M-6-12.north east);%
			
			\draw[densely dashed, draw=lightgray] (M-2-5.south east) -- (M-2-5.north east);			
			\draw[densely dashed] (M-3-7.south west) -- (M-4-7.west); 
			\draw[densely dashed] (M-6-10.north west) -- (M-6-9.north); 
		\end{pgfonlayer} 
		
		\draw [decorate,decoration={brace,amplitude=0.5em,raise=1pt,mirror},gray, left] ([yshift=-1pt]M-1-5.north west) -- ([yshift=1pt]M-2-5.south west) node[gray, midway, left, xshift=-3.5pt] {$k_1$};
 		\draw [decorate,decoration={brace,amplitude=0.5em,raise=1pt,mirror},gray, left] ([yshift=-1pt]M-2-5.south east) -- ([yshift=1pt]M-3-5.south east) node[gray, midway, left, xshift=-3.5pt] {$k_2$};
 		\draw [decorate,decoration={brace,amplitude=0.5em,raise=1pt,mirror},gray, left] ([yshift=-4pt]M-5-9.south east) -- ([yshift=1pt]M-6-9.south east) node[gray, midway, left, xshift=-3.5pt] {$k_v$};
 		\draw [decorate,decoration={brace,amplitude=0.5em,raise=1pt},gray,above] (M-1-5.north west) -- (M-1-5.north east) node[gray, midway, yshift=4pt] {$n_1$};
 		\draw [decorate,decoration={brace,amplitude=0.5em,raise=1pt},gray,above] (M-1-5.north east) -- (M-1-7.north west) node[gray, midway, yshift=4pt] {$n_2$};
 		\draw [decorate,decoration={brace,amplitude=0.5em,raise=1pt},gray,above] (M-1-9.north east) -- (M-1-12.north east) node[gray, midway, yshift=4pt] {$n_v$};

\end{tikzpicture}
}